\documentclass[12pt]{amsart}

\usepackage{tikz}
\usetikzlibrary{positioning}
\usepackage[all]{xy}
\usepackage{amsmath}
\usepackage{hyperref}

\newtheorem{theorem}{Theorem}[section]
\newtheorem{lemma}[theorem]{Lemma}

\theoremstyle{definition}
\newtheorem{definition}[theorem]{Definition}
\newtheorem{example}[theorem]{Example}

\theoremstyle{remark}
\newtheorem{remark}[theorem]{Remark}

\numberwithin{equation}{section}

\DeclareMathOperator{\Map}{Map}

\DeclareMathOperator{\Id}{Id}

\renewcommand{\vec}{\overrightarrow}

\begin{document}

\title{Simplicial models for concurrency}

\author{Peter Bubenik}
\address{Department of Mathematics, Cleveland State University, 2121 Euclid Ave RT 1515, Cleveland, OH 44115, USA}
\curraddr{}
\email{p.bubenik@csuohio.edu}
\thanks{}

\subjclass[2010]{68Q85, 55U10, 18D20, 55U35}

\keywords{Concurrency; Simplicial set; Simplicial category; Necklace}

\date{\today}


\begin{abstract}
We model both concurrent programs and the possible executions from one state to another in a concurrent program using simplices. The latter are calculated using necklaces of simplices in the former. 
\end{abstract}

\maketitle

\section{Introduction}
\label{sec:introduction}

We develop a discrete model for processes that is useful in the concurrent setting. It is based on the traditional model given by graphs. In this model, vertices of the graph represent states and edges represent transitions. The edges may be directed or undirected. 

A more sophisticated model appropriate for concurrency is given by simplicial complexes. For example, the following triangle,
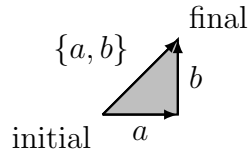
\begin{center}
\begin{tikzpicture}[thick]
  \filldraw [fill=lightgray] (0,0) node [anchor=north east]
  {initial} -- node [anchor=north] {$a$} (1,0) -- node [anchor=west]
  {$b$} (1,1) node [anchor = south west] {final} -- node [anchor=south
  east] {$\{a,b\}$} (0,0);
  \pgfsetarrowsend{latex}
  \draw (0,0) -- (1,0);
  \draw (0,0) -- (1,1);
  \draw (1,0) -- (1,1);
\end{tikzpicture}
\end{center}
models two processes $a$ and $b$ where the boundary models two executions, either $a$ followed by $b$, or $a$ and $b$ concurrently. The interior of the triangle models intermediate executions. Adding another triangle we obtain a simplicial complex that models all possible executions of $a$ and $b$. 
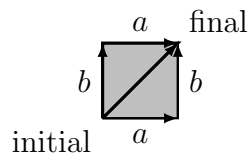
\begin{center}
\begin{tikzpicture}[thick]
  \filldraw [fill=lightgray] (0,0) node [anchor=north east] {initial} -- node [anchor=north] {$a$} (1,0) -- node [anchor=west] {$b$} (1,1) node [anchor = south west] {final} -- (0,0);
  \filldraw[fill=lightgray,very thick] (0,0) -- node [anchor=east]
  {$b$} (0,1) -- node [anchor=south] {$a$} (1,1) -- (0,0);
  \pgfsetarrowsend{latex}
  \draw (0,0) -- (1,0);
  \draw (0,0) -- (1,1);
  \draw (1,0) -- (1,1);
  \draw (0,0) -- (0,1);
  \draw (0,1) -- (1,1);
\end{tikzpicture}
\end{center}
For three processes, the analogous model is a cube that is subdivided into six tetrahedra.

The combinatorial versions of simplicial complexes are simplicial sets. These consist of sets of abstract vertices, edges, triangles, tetrahedra, and higher dimensional simplices. For a simplicial set, the corresponding simplicial complex is its geometric realization. Our figures will consist of simplicial complexes, but we intend these to represent the corresponding simplicial set.

Using Dijkstra's formalism for concurrent programs~\cite{dijkstra:book}, we
give an explicit construction of a simplicial model for a concurrent
program using filtered simplicial sets.
The vertices of this simplicial model give states of the concurrent program. 
The executions between states are modeled by paths between vertices in
the simplicial model.
There is a growing literature of related models, which are mostly
continuous. For some examples,
see~\cite{frg:atc,grandis:dht1,bubenik:vanKampen,bubenikWorytkiewicz:modelCategoryForLPS,krishnan:streams,kahl:reducingCubical,fghr:components}.

In concurrency, understanding the possible execution paths is one of
the main sources of difficulty.  In a simplicial set, the equivalence
classes of paths can be described by a category called the fundamental
category~\cite{joyal:qcatAndKan} or the path
category~\cite{jardine:pathCategories}. In this category the objects
are the vertices of the simplicial set, and between objects there is a
set of equivalence classes of paths. 

We consider a related but more sophisticated construction in which the
objects are vertices of the simplicial set, but the executions from
one state to another are described by a simplicial set. That is, our
model is a simplicial category (i.e., a category enriched over
simplicial sets). This construction can detect higher order structure
that is invisible to the path category. To calculate it, we consider
necklaces of simplices in our simplicial set models. In a related
paper~\cite{raussen:simplicialModels}, Raussen constructs simplicial
models of execution spaces for continuous models.

\emph{Outline of the paper:} In Section~\ref{sec:machinery}, we discuss simplicial sets, filtered simplicial sets, necklaces, simplicial categories, and some of the other necessary mathematical machinery. In Section~\ref{sec:concurrency}, we apply these constructions to construct simplicial models of concurrent systems and also of their execution spaces. Our examples include an example in which the structure of the executions is not detected by the path category, but is captured by our methods. 
In Section~\ref{sec:futureWork}, we discuss a suitable model structure for this setting and we remark that one may try to apply these methods using cubical sets instead of simplicial sets.

\section{Mathematical machinery}
\label{sec:machinery}

\subsection{Simplicial sets}
\label{sec:simplicial-sets}

A \emph{simplicial set}, $X$, consists of a sequence of sets, $\{X_0, X_1, X_2,
\ldots\}$ together with \emph{face} maps $d_i: X_k \to X_{k-1}$ for $0 \leq i
\leq k$ and \emph{degeneracy} maps $s_i: X_k \to X_{k+1}$ for $0 \leq i \leq
k$ satisfying the following \emph{simplicial identities}. 
\begin{align*}
d_id_j& = d_{j-1}d_i \text{ if } i<j\\
s_is_j& = s_{j+1}s_i  \text{ if } i\leq j\\
d_is_j& = s_{j-1}d_i  \text{ if } i<j\\
d_js_j& = \Id = d_{j+1}s_j\\
d_is_j& = s_jd_{i-1}  \text{ if } i>j+1
\end{align*}
The elements of $X_k$ are
called \emph{$k$-simplices}.
The $0$-simplices are called \emph{vertices} and the $1$-simplices are called \emph{edges}. A morphism of simplicial sets, $f:X \to Y$ consists of a sequence of functions $f_k:X_k \to Y_k$ that commute with the face and degeneracy maps. The geometric realization of a simplicial set is a simplicial complex.

For example, for each $n$ there is a simplicial set $\Delta^n$ whose
geometric realization is the standard geometric
$n$-simplex. $\Delta^0$ is the simplicial set with one element in each
$X_k$ and all of the face and degeneracy maps given by the
identity. The $k$-simplices that are not in the image of a degeneracy
map are called \emph{nondegenerate}. Thus, $\Delta^0$ has only one nondegenerate simplex, which is a $0$-simplex. The only nondegenerate simplices in $\Delta^1$ are $a,b \in \Delta^1_0$ and $e \in \Delta^1_1$ where $d_0(e)=b$ and $d_1(e)=a$.
Note that the initial vertex is $d_1(e)$ and the final vertex is $d_0(e)$.

Simplicial sets and their morphisms form a category. Geometric
realization is a functor from this category to the category of
topological spaces.  A more elegant but more sophisticated way of
defining simplicial sets is as contravariant functors from the
category of finite ordinals and order preserving maps to the category
of sets.


\subsection{Filtered simplicial sets}
\label{sec:filt-simpl-sets}

A sub-simplicial set $A \subseteq X$ is a sequence of subsets $A_i
\subseteq X_i$ that are closed under the restrictions of the face and
degeneracy maps.
A \emph{filtration} on a simplicial set $X$ is an increasing sequence
$X_{(0)} \subseteq X_{(1)} \subseteq X_{(2)} \subseteq \ldots$ of 
sub-simplicial sets of $X$. The simplices in $X_{(d)}$ are said to be
in degree $d$.
A \emph{$m$-filtration} on $X$ is a collection of sub-simplicial sets
$X_{(i_1,\ldots,i_m)}$ of $X$ such that if $i_j \leq i'_j$ for $1\leq j \leq m$
then $X_{(i_1,\ldots,i_m)}$ is a sub-simplicial set of $X_{(i'_1,\ldots,i'_m)}$.

For simplicial sets $X$ and $Y$, their product $X \times Y$ is given by $(X\times Y)_n
= X_n \times Y_n$ with face and degeneracy maps $d_i = d_i \times d_i$
and $s_i = s_i \times s_i$. Note that the nondegenerate simplices of
$X \times Y$ are not necessarily products of nondegenerate simplices
of $X$ and $Y$. A standard example is $\Delta^1 \times \Delta^1$ which
is a triangulation of the square.
For filtered simplicial sets $X$ and $Y$, $X \times Y$ has an induced
filtration where  $(X \times Y)_{(k)}$ is the union of the
sub-simplicial sets $X_{(i)} \times
Y_{(j)}$ where $i+j = k$.
Similarly, products of $m$-filtered simplicial sets have an induced $m$-filtration.

\subsection{The path category of a simplicial set}
\label{sec:path-categ-simpl}

Given a simplicial set $X$, we can define the {path category}, $P(X)$,
as follows. For a thorough exposition
see~\cite{jardine:pathCategories}.
The path category is also called the fundamental category.

\begin{definition}
  Let the \emph{path category} $P(X)$ of the simplicial set $X$ be the category whose objects are the vertices of $X$ and whose morphisms are concatenations of edges in $X$ modulo the equivalence relation generated by the following relations
\begin{gather*}
  s_0(a) \sim \Id_a, \quad \text{for all } a \in X_0, \text{ and}\\
  d_1(t) \sim d_2(t)d_0(t), \quad \text{for all } t \in X_2.
\end{gather*}
\end{definition}

\subsection{The simplicial category of an ordered simplicial set}
\label{sec:simpl-categ}

Since the path category only depends on the $k$-simplices for $k\leq 2$, it does not detect higher order structure. 

For this purpose, we consider the following construction of a category
enriched over simplicial sets. In such a category, between two objects,
instead of set of morphisms, we have a simplicial set of
morphisms. For brevity, this is also called a simplicially enriched
category, or a simplicial category.
For a thorough exposition of the simplicial category of a simplicial
set see~\cite{duggerSpivak:rigidification}.

A \emph{necklace} is a simplicial set, $T$, of the form 
\begin{equation*}
  \Delta^{n_1} \vee \Delta^{n_2} \vee \ldots \vee \Delta^{n_k}
\end{equation*}
where the final vertex of $\Delta^{n_i}$ is glued to the initial vertex of $\Delta^{n_{i+1}}$.
The simplicial set $\Delta^{n_i}$ is called a \emph{bead} of $T$. Vertices of $T$ that are initial or final vertices of any of the beads are called \emph{joints} of $T$. The set of joints of $T$ is denoted $J_T$. Call the initial vertex of $\Delta^{n_1}$ the initial vertex of $T$ and the final vertex of $\Delta^{n_k}$ the final vertex of $T$.

A \emph{flag} of $T$ is an increasing sequence $\vec{T} = (T^0 \subseteq \cdots \subseteq T^n)$ of subsets of vertices of $T$. The \emph{length} of $\vec{T}$ is $n$. A flag is \emph{flanked} if $T^0 = J_T$ and $T^n = T_0$.

A simplicial set $X$ is \emph{ordered} if for each $a \in X_0$,
$P(X)(a,a)$ consists of only the identity morphism, and no two
simplices have the same set of vertices.  The simplicial sets arising
in our models in Section~\ref{sec:concurrency} will all be
ordered. One can define the simplicial category for a simplicial set
without this condition, but for ordered simplicial sets we have the
following nice explicit construction
from~\cite{duggerSpivak:rigidification}.

\begin{definition}
  For an ordered simplicial set $X$, let $\mathfrak{C}(X)$ be the category enriched in simplicial sets whose objects are vertices in $X$, and such that for vertices $a$ and $b$, $\mathfrak{C}(X)(a,b)$ is the simplicial set given as follows.
Let $\mathfrak{C}(X)(a,b)_n$ be the set of triples $(T,f,\vec{T})$ where $T$ is a necklace, $f:T \to X$ is a injective map of simplicial sets such that the images of the initial and final vertices of $T$ are $a$ and $b$ respectively, and $\vec{T}$ is a flanked flag of length $n$.
The degeneracy maps, $s_i$, for $0\leq i \leq n$ are given by 
\begin{equation*}
  s_i((T,f,T^0 \subseteq \cdots \subseteq T^n)) = (T,f,T^0 \subseteq \cdots \subseteq T^i \subseteq T^i \subseteq \cdots T^n).
\end{equation*}
The face maps, $d_i$, for $0 < i < n$ are given by 
\begin{equation*}
  d_i((T,f,T^0 \subseteq \cdots \subseteq T^n)) = (T,f,T^0 \subseteq \cdots \subseteq T^{i-1} \subseteq T^{i+1} \subseteq \cdots T^n).
\end{equation*}
For $i=0$, let $T'$ the unique subnecklace of $T$ whose set of joints is $T^1$ and whose set of vertices is $T^n$.
Similarly, for $i=n$, let $T''$ be the unique subnecklace of $T$ whose set of joints is $T^0$, and whose set of vertices is $T^{n-1}$. 
Then
\begin{gather*}
  d_0((T,f,T^0 \subseteq \cdots \subseteq T^n)) = (T',f,T^1 \subseteq \cdots \subseteq T^n), \text{ and}\\
  d_n((T,f,T^0 \subseteq \cdots \subseteq T^n)) = (T'',f,T^0 \subseteq \cdots \subseteq T^{n-1}).
\end{gather*}
\end{definition}

\section{Concurrency}
\label{sec:concurrency}

We use Dijkstra's abstraction of concurrent programming in which \emph{semaphores} are used to control access by multiple processes to a common resource ~\cite{dijkstra:book}. More precisely, a resource that may be used by only $k$ processes at once is controlled by a $k$-semaphore. This is simply a nonnegative counter which starts at $k$. When a process wants to use the resource it tries to decrement the counter and when it is finished using the resource it increments the counter. $1$-semaphores are also called \emph{binary semaphores}.

A concurrent program can be abstractly written as a sequence of
operations on semaphores. Following Dijkstra's original notation, for
a semaphore $a$, let $Pa$ denote decrementing $a$ and let $Va$
denote incrementing $a$.
We consider programs that are finite sequences of these operations.
Abstractly a program for one process is given by
\begin{equation} \label{eq:PVprogram}
  P = O_1a_{i_1}.O_2a_{i_2}.\cdots.O_Na_{i_N},
\end{equation}
where $O_j$ is either operator $P$ or $V$ and $a_{i_j}$ is one of the semaphores. Let $s_0$ be the initial state of the program $P$ and let $s_j$ be state of the program after the operation $O_ja_{i_j}$.

\subsection{Simplicial models of concurrent programs}
\label{sec:simplicialModels}

\begin{definition} \label{def:modelProcess}
The \emph{simplicial model} for the program $P$ for a single process
in~\eqref{eq:PVprogram} is modeled by the simplicial set $X$ that is the
necklace of $N$ $1$-simplices.
The vertices of the necklace correspond to states of the program.
If there are $m$ shared resources $a_1, \ldots, a_m$ we define an $m$-filtration on this
necklace. The initial vertex has degree $(0,\ldots, 0)$. 
The operation $Pa_i$ increases by one the $a_i$ degree of the corresponding edge and the
subsequent vertices and edges in the necklace. 
The operation $Va_i$ decreases by one the $a_i$ degree of the vertices and
edges in the necklace following the corresponding edge.
\end{definition}

\begin{example} \label{example:swissA}
  Let $a$ and $b$ be two binary semaphores. Consider the program
  \begin{equation*}
    Pa.Pb.Va.Vb
  \end{equation*}
  This program can be modeled by the following bifiltered simplicial set.
 
  
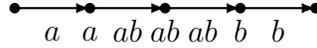
\begin{center}
   \begin{tikzpicture}[thick]
      \pgfsetarrows{latex-latex}
      \pgfsetarrowsstart{}
   \foreach \x in {0,4}
         \fill (\x,0) circle (2pt);
   \foreach \x in {1}
         \fill (\x,0) circle (2pt) node[below=3pt] {$a$};
   \foreach \x in {3}
         \fill (\x,0) circle (2pt) node[below] {$b$};
   \foreach \x in {2}
         \fill (\x,0) circle (2pt) node[ below] {$ab$};
   \foreach \x in {0}
       \draw (\x,0) -- node[below=3pt] {$a$} (\x+1,0) ;
   \foreach \x in {3}
       \draw (\x,0) -- node[below] {$b$} (\x+1,0) ;
   \foreach \x in {1,2}
       \draw (\x,0) -- node[below] {$ab$} (\x+1,0) ;
 \end{tikzpicture}
\end{center}
Here the unlabeled vertices and edges are in degree $(0,0)$. The
vertices and edges labeled $a$, $b$ and $ab$ are in degree
$(1,0)$, $(0,1)$ and $(1,1)$ respectively.
\end{example}

Assume programs $P_1$, \ldots, $P_n$ with shared resources $a_1,\ldots,a_m$ have as simplicial models the $m$-filtered simplicial sets $X_1$, \ldots $X_n$. 
We want a model for the concurrent program $P = (P_1 | \cdots | P_n)$.
The simplicial set $X_1 \times  \cdots \times X_n$ has an induced $m$-filtration (see Section~\ref{sec:filt-simpl-sets}). 
A \emph{state} of $P$ is a vertex $x = (x_1,\ldots,x_n) \in X_1 \times
\cdots \times X_n$.

\begin{definition} \label{def:validState}
  We call a state $x = (x_1,\ldots,x_n)$ of $P$ \emph{valid} if across
  $1\leq i \leq n$ the uses of resources by $P_i$ at $s_{j_i}$ are
  compatible, where $s_{j_i}$ is the state of $P_i$ corresponding to the vertex $x_i$ of  
  the model $X_i$.
  More precisely, for $1\leq i \leq n$, $1 \leq j \leq m$ and $x_i \in
  P_i$, let $D_{ij}(x_i)$ be the amount that $a_j$ has decreased from
  $k_j$ if the program for the process $P_i$ runs from its initial
  state to the state $x_i$.
  A state $x=(x_1,\ldots,x_n)$ is \emph{valid} if for each $1\leq j
  \leq m$
  \begin{equation*}
    \sum_{i=1}^n D_{ij}(x_i) \leq k_j.
  \end{equation*}
\end{definition}

\begin{lemma}
  The set of valid states of $P$ is given by the vertices in $(X_1 \times \cdots \times X_n)_{(k_1,\ldots,k_m)}$.
\end{lemma}

\begin{proof}
  Consider $x = (x_1, \ldots, x_n) \in X_1 \times \cdots \times X_n$. 
  Let $1\leq i \leq n$ and $1\leq j \leq m$.
  From the definition of the grading on $X_i$, $x_i$ has $a_j$ degree
  $D_{ij}(X_i)$.
  So the $a_j$ degree of $x$ is $\sum_{i=1}^n D_{ij}(x_i)$.
  Thus the degree of $x$ is 
 \begin{equation*}
    \left(\sum_{i=1}^n D_{i1}(x_i),\ldots,\sum_{i=1}^n
    D_{im}(x_i)\right).
  \end{equation*}
 Therefore $x$ is valid if and only if $x \in (X_1 \times \cdots
  \times X_n)_{(k_1,\ldots,k_m)}$.  
\end{proof}

An edge of $X_1 \times \cdots \times X_n$ is of the form $(e_1,\ldots,e_n)$ where $e_i$ is an edge of $X_i$. Note that $e_i$ may be degenerate. That is $e_i = sx_i$ where $x_i$ is a vertex of $X_i$.

\begin{definition} \label{def:validEdge}
  Define an edge $e = (e_1,\ldots,e_n)$ to be \emph{valid} if across $1\leq i \leq n$ the uses of resources by $P_i$ from $s_{j_i}$ to $s_{j'_i}$ are compatible, where $s_{j_i}$ and $s_{j'_i}$ are the states corresponding to $d_1(e_i)$ and $d_0(e_i)$, respectively. More precisely, using the notation of Definition~\ref{def:validState}, $e$ is valid if for each $1\leq j\leq m$,
  \begin{equation*}
    \sum_{i=1}^n \max\{D_{ij}(d_1(e_i)), D_{ij}(d_0(e_i))\} \leq k_j
  \end{equation*}
\end{definition}

\begin{lemma}
  The set of valid edges of $P$ is given by the edges in $(X_1 \times \cdots \times X_n)_{(k_1,\ldots,k_m)}$.
\end{lemma}

\begin{proof}
  Consider $e=(e_1,\ldots,e_n) \in X_1 \times \cdots \times X_n$. Let $1\leq i\leq n$ and $1\leq j \leq m$. From the definition of the grading on $X_i$, $e_i$ has $a_j$ degree equal to the maximum of the $a_j$ degree of its vertices $d_1(e_i)$ and $d_0(e_i)$.
  Thus the degree of $e$ is 
  \small
  \begin{equation*}
    \left(\sum_{i=1}^n \max\{D_{i1}(d_1(e_i)),D_{i1}(d_0(e_i))\},\ldots,\sum_{i=1}^n
    \max\{D_{im}(d_1(e_i)),D_{im}(d_0(e_i))\}\right).
  \end{equation*}
  \normalsize
  Therefore $x$ is valid if and only if $x \in (X_1 \times \cdots
  \times X_n)_{(k_1,\ldots,k_m)}$.  
\end{proof}

\begin{definition} \label{def:modelProgram}
  Assume the program $P = (P_1 | \ldots | P_n)$ has shared resources
  $a_1, \ldots, a_m$ where $a_i$ is a $k_i$-semaphore.
  The \emph{simplicial model} for $P$ is given by the simplicial set 
  \begin{equation*}
    (X_1 \times \cdots \times X_n)_{(k_1, \ldots,k_m)}.
  \end{equation*}
\end{definition}

\begin{example} \label{example:swissReduced}
  Let $a$ and $b$ be two binary semaphores. Consider the two programs
  \begin{equation*}
    A = Pa.Pb.Vb.Va \quad \text{and} \quad B = Pb.Pa.Va.Vb
  \end{equation*}
  These can be modeled by the following bifiltered simplicial sets $\overline{A}$ and $\overline{B}$.
 
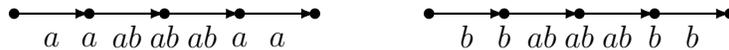
  \begin{center}
   \begin{tikzpicture}[thick]
      \pgfsetarrows{latex-latex}
      \pgfsetarrowsstart{}
   \foreach \x in {0,4}
         \fill (\x,0) circle (2pt);
   \foreach \x in {1,3}
         \fill (\x,0) circle (2pt) node [below=3pt] {$a$};
  \foreach \x in {2}
         \fill (\x,0) circle (2pt) node [anchor=north] {$ab$};
   \foreach \x in {0,3}
       \draw (\x,0) -- node [below=3pt] {$a$} (\x+1,0) ;
      \foreach \x in {1,2}
       \draw (\x,0) -- node [anchor=north] {$ab$} (\x+1,0) ;
  \end{tikzpicture}
  \quad \quad  \quad
   \begin{tikzpicture}[thick]
      \pgfsetarrows{latex-latex}
      \pgfsetarrowsstart{}
   \foreach \x in {0,4}
         \fill (\x,0) circle (2pt);
   \foreach \x in {1,3}
         \fill (\x,0) circle (2pt) node [anchor=north] {$b$};
  \foreach \x in {2}
         \fill (\x,0) circle (2pt) node [anchor=north] {$ab$};
   \foreach \x in {0,3}
       \draw (\x,0) -- node [anchor=north] {$b$} (\x+1,0) ;
      \foreach \x in {1,2}
       \draw (\x,0) -- node [anchor=north] {$ab$} (\x+1,0) ;
  \end{tikzpicture}
\end{center}
Here the unlabeled vertices and edges are in filtration $(0,0)$. The
vertices and edges labeled $a$, $b$ and $ab$ are in filtration
$(1,0)$, $(0,1)$ and $(1,1)$ respectively.

 Now consider the concurrent execution of $A$ and $B$, denoted $(A|B)$. It is modeled by the following simplicial set $(\overline{A} \times \overline{B})_{(1,1)}$. 

  
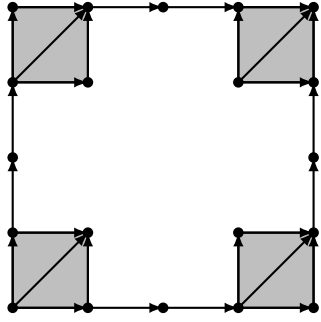
\begin{center}
    \begin{tikzpicture}[thick]
   \pgfsetarrowsend{}
   \foreach \x in {0,3}
     \foreach \y in {0,3}
       \draw [fill=lightgray] (\x,\y) -- (\x+1,\y) -- (\x+1,\y+1) -- (\x,\y+1) -- (\x,\y);
      \pgfsetarrowsend{latex} 
     \foreach \x in {0,4} 
     \foreach \y in {0,1,2,3,4}
            \fill (\x,\y) circle (2pt);
     \foreach \x in {1,2,3} 
     \foreach \y in {0,4}
            \fill (\x,\y) circle (2pt);
     \fill (1,1) circle (2pt);
     \fill (1,3) circle (2pt);
     \fill (3,1) circle (2pt);
     \fill (3,3) circle (2pt);
   \foreach \x in {0,1,2,3}
      {
      \draw (\x,4) -- (\x+1,4);
       \draw (\x,0) -- (\x+1,0);
    }
    \foreach \y in {0,1,2,3}
    {
      \draw (0,\y) -- (0,\y+1);
      \draw (4,\y) -- (4,\y+1);
    }
    \foreach \x in {0,3}
       \foreach \y in {0,3}
         \draw (\x,\y) -- (\x+1,\y+1);
     \draw (1,0) -- (1,1);
     \draw (0,1) -- (1,1);
     \draw (3,0) -- (3,1); 
     \draw (3,1) -- (4,1);
     \draw (0,3) -- (1,3);
     \draw (1,3) -- (1,4);
     \draw (3,3) -- (3,4);
     \draw (3,3) -- (4,3);
\end{tikzpicture}
  \end{center}
\end{example}

\begin{example} \label{example:cube}
  Let $a$ be a $2$-semaphore. Consider the three identical programs,
  $A$, $B$ and $C$ given by
  \begin{equation*}
    Pa.Va
  \end{equation*}
  This program is modeled by the following filtered simplicial set $X$.
 
  \begin{center}
    \begin{tikzpicture}[thick]
      \pgfsetarrows{latex-latex}
      \pgfsetarrowsstart{}
      \foreach \x in {0,1}
       \draw (\x,0) -- node [anchor=north] {$a$} (\x+1,0) ;
     \foreach \x in {1}
         \fill (\x,0) circle (2pt) node [anchor=north] {$a$};
     \foreach \x in {0,2}
         \fill (\x,0) circle (2pt);
  \end{tikzpicture}
  \end{center}

  The concurrent program $(A|B|C)$ is modeled by the simplicial set
  $(X \times X \times X)_{(2)}$, which is a triangulation of the boundary of the cube.
\end{example}

\begin{remark}
  One can also define analogous cubical models for concurrent
  programs. The simplicial models presented here are the
  triangulations of those cubical models.
\end{remark}

\subsection{Simplicial models of execution spaces}
\label{sec:necklaces}

Next we construct simplicial models for the space of executions from one state to another in the simplicial models in Section~\ref{sec:simplicialModels}.

\begin{definition}
 For vertices $a$ and $b$ in a simplicial model $X$, we
 define the \emph{simplicial model} of the execution space from $a$ to $b$ to be the simplicial set $\mathfrak{C}(X)(a,b)$.
\end{definition}

\begin{example} \label{example:cubeNecklaces}
  Let $X$ be the boundary of the cube 
labeled as follows.
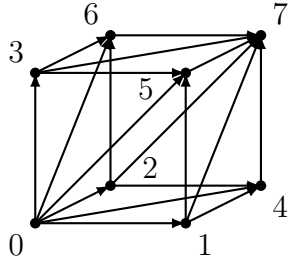
\begin{center}
  \begin{tikzpicture}[thick,scale=2]
    \pgfsetarrowsend{latex}
    \def\xdeep{0.5}
    \def\ydeep{0.25}
   \foreach \x in {0,1}
   {
     \draw (\x,0) -- (\x,1);
     \draw (\x+\xdeep,\ydeep) -- (\x+\xdeep,1+\ydeep);
     \draw (\x,0) -- (\x+\xdeep,1+\ydeep);
     \foreach \y in {0,1}
      \draw (\x,\y) -- (\x+\xdeep,\y+\ydeep);
 }
  \foreach \y in {0,1}
  {
    \draw (0,\y) -- (1,\y);
    \draw (\xdeep,\y+\ydeep) -- (1+\xdeep,\y+\ydeep);
    \draw (0,\y) -- (1+\xdeep,\y+\ydeep);
 }
  \draw (0,0) -- (1,1);
  \draw (\xdeep,\ydeep) -- (1+\xdeep, 1+\ydeep);
  \fill (0,0) circle (1pt) node [anchor=north east] {$0$};
  \fill (1,0) circle (1pt) node [anchor=north west] {$1$};
  \fill (\xdeep,\ydeep) circle (1pt) node [above right=-0.25ex and 1.5ex] {$2$};
 \fill (0,1) circle (1pt) node [anchor=south east] {$3$};
  \fill (1+\xdeep,\ydeep) circle (1pt) node [anchor=north west] {$4$};
  \fill (1,1) circle (1pt) node [below left=-0.4ex and 1.5ex] {$5$};
  \fill (\xdeep,1+\ydeep) circle (1pt) node [anchor=south east] {$6$};
  \fill (1+\xdeep,1+\ydeep) circle (1pt) node [anchor=south west] {$7$};
\end{tikzpicture}
\end{center}

Then $\Map_X(0,7)$ is given by the following simplicial set.

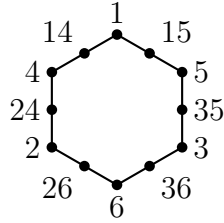
\begin{center}
 \begin{tikzpicture}[thick,scale=1]
    \pgfmathparse{cos(90)}
     \let\xa\pgfmathresult
     \pgfmathparse{sin(90)}
     \let\ya\pgfmathresult
     \pgfmathparse{cos(150)}
     \let\xd\pgfmathresult
     \pgfmathparse{sin(150)}
     \let\yd\pgfmathresult
     \pgfmathparse{cos(210)}
     \let\xb\pgfmathresult
     \pgfmathparse{sin(210)}
     \let\yb\pgfmathresult
     \pgfmathparse{cos(270)}
     \let\xf\pgfmathresult
     \pgfmathparse{sin(270)}
     \let\yf\pgfmathresult
     \pgfmathparse{cos(330)}
     \let\xc\pgfmathresult
     \pgfmathparse{sin(330)}
     \let\yc\pgfmathresult
     \pgfmathparse{cos(30)}
     \let\xe\pgfmathresult
     \pgfmathparse{sin(30)}
     \let\ye\pgfmathresult
    \fill (\xa, \ya) circle (2pt) node [anchor=south] {$1$};
    \fill (\xb, \yb) circle (2pt) node [anchor=east] {$2$};
    \fill (\xc, \yc) circle (2pt) node [anchor=west] {$3$};
    \fill (\xd, \yd) circle (2pt) node [anchor=east] {$4$};
    \fill (\xe, \ye) circle (2pt) node [anchor=west] {$5$};
    \fill (\xf, \yf) circle (2pt) node [anchor=north] {$6$};
    \draw (\xa,\ya) -- (\xd,\yd) -- (\xb,\yb) -- (\xf,\yf) -- (\xc,\yc) -- (\xe,\ye) -- (\xa,\ya);
    \fill (\xa/2+\xd/2,\ya/2+\yd/2) circle (2pt) node  [anchor=south east] {$14$};
    \fill (\xb/2+\xd/2,\yb/2+\yd/2) circle (2pt) node  [anchor=east] {$24$};
    \fill (\xb/2+\xf/2,\yb/2+\yf/2) circle (2pt) node  [anchor=north east] {$26$};
    \fill (\xf/2+\xc/2,\yf/2+\yc/2) circle (2pt) node [anchor=north west] {$36$};
    \fill (\xe/2+\xc/2,\ye/2+\yc/2) circle (2pt) node [anchor=west] {$35$};
    \fill (\xe/2+\xa/2,\ye/2+\ya/2) circle (2pt) node [anchor=south west] {$15$};
\end{tikzpicture}
\end{center}
Here a vertex labeled $i$ represents the necklace with flanked flag $(\Delta^1 \vee \Delta^1, \{0,i,7\})$ and a vertex labeled $ij$  represents the necklace with flanked flag $(\Delta^1 \vee \Delta^1 \vee \Delta^1, \{0,i,j,7\})$. An edge between vertices $i$ and $ij$ represents the necklace with flanked flag $( \Delta^2 \vee \Delta^1, \{0,i,7\} \subset \{0,i,j,7\})$ or 
$( \Delta^1 \vee \Delta^2, \{0,i,7\} \subset \{0,i,j,7\})$.

Thus $\Map_X(0,7)$ is homotopy equivalent to the circle. We remark that this is an example of higher order structure that is not detected by the path category~\cite{jardine:pathCategories}.
\end{example}

\begin{example} \label{example:petriNet}
 Let $X$ be the following simplicial set. For clarity, we omit the $1$-simplices along the diagonals of each of the squares.
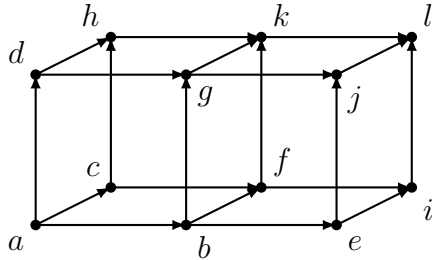
\begin{center}
  \begin{tikzpicture}[thick,scale=2]
    \pgfsetarrowsend{latex}
    \def\xdeep{0.5}
    \def\ydeep{0.25}
   \foreach \x in {0,1,2}
   {
     \draw (\x,0) -- (\x,1);
     \draw (\x+\xdeep,\ydeep) -- (\x+\xdeep,1+\ydeep);
     \foreach \y in {0,1}
      \draw (\x,\y) -- (\x+\xdeep,\y+\ydeep);
 }
  \foreach \x in {0,1}
  \foreach \y in {0,1}
  {
    \draw (\x,\y) -- (\x+1,\y);
    \draw (\x+\xdeep,\y+\ydeep) -- (\x+1+\xdeep,\y+\ydeep);
 }
  \fill (0,0) circle (1pt) node [anchor=north east] {$a$};
  \fill (1,0) circle (1pt) node [anchor=north west] {$b$};
  \fill (\xdeep,\ydeep) circle (1pt) node [anchor=south east] {$c$};
 \fill (0,1) circle (1pt) node [anchor=south east] {$d$};
 \fill (2,0) circle (1pt) node [anchor=north west] {$e$};
\fill (1+\xdeep,\ydeep) circle (1pt) node [anchor=south west] {$f$};
  \fill (1,1) circle (1pt) node [anchor=north west] {$g$};
 \fill (\xdeep,1+\ydeep) circle (1pt) node [anchor=south east] {$h$};
\fill (2+\xdeep,\ydeep) circle (1pt) node [anchor=north west] {$i$};
  \fill (2,1) circle (1pt) node [anchor=north west] {$j$};
 \fill (1+\xdeep,1+\ydeep) circle (1pt) node [anchor=south west] {$k$};
 \fill (2+\xdeep,1+\ydeep) circle (1pt) node [anchor=south west] {$l$};
\end{tikzpicture}
\end{center}

Then $\Map_X(a,l)$ is given by the following simplicial set.

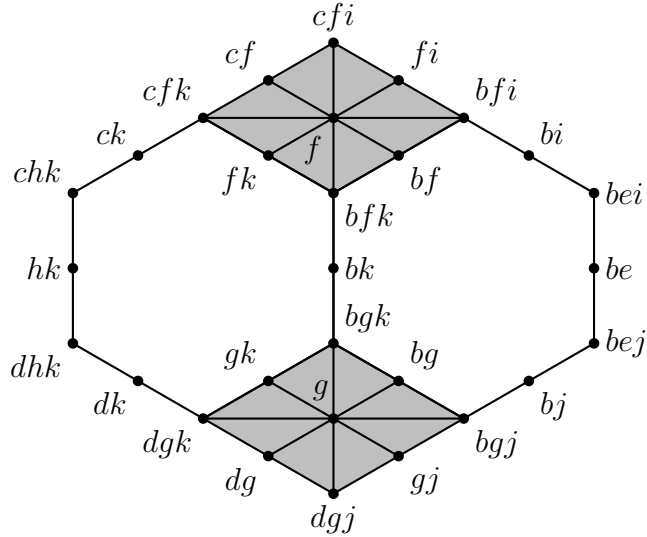
\begin{center}
 \begin{tikzpicture}[thick,scale=2]
    \pgfmathparse{cos(90)}
     \let\xa\pgfmathresult
     \pgfmathparse{sin(90)}
     \let\ya\pgfmathresult
     \pgfmathparse{cos(150)}
     \let\xd\pgfmathresult
     \pgfmathparse{sin(150)}
     \let\yd\pgfmathresult
     \pgfmathparse{cos(210)}
     \let\xb\pgfmathresult
     \pgfmathparse{sin(210)}
     \let\yb\pgfmathresult
     \pgfmathparse{cos(270)}
     \let\xf\pgfmathresult
     \pgfmathparse{sin(270)}
     \let\yf\pgfmathresult
     \pgfmathparse{cos(330)}
     \let\xc\pgfmathresult
     \pgfmathparse{sin(330)}
     \let\yc\pgfmathresult
     \pgfmathparse{cos(30)}
     \let\xe\pgfmathresult
     \pgfmathparse{sin(30)}
     \let\ye\pgfmathresult
    \foreach \y in {-\ya,\ya}
    {
      \draw [fill=lightgray] (0,\y-\ya+\ye) -- (\xe,\y) -- (0,\y+\ya-\ye) -- (-\xe,\y) -- (0,\y-\ya+\ye);
      \draw (0,\y+\ya-\ye) -- (0,\y-\ya+\ye);
      \draw (-\xe,\y) -- (\xe,\y);
      \draw (-\xe/2,\y-\ye/2) -- (\xe/2,\y+\ye/2);
      \draw (-\xe/2,\y+\ye/2) -- (\xe/2,\y-\ye/2);
   }
     \foreach \x in {-\xe,\xe}
     {
   \draw (\x+\xa,\ya) -- (\x+\xd,\yd) -- (\x+\xb,\yb) -- (\x+\xf,\yf) -- (\x+\xc,\yc) -- (\x+\xe,\ye) -- (\x+\xa,\ya);
   }
    \foreach \x in {-\xe}
    {
    \fill (\x+\xa, \ya) circle (1pt) node [anchor=south east] {$cfk$};
    \fill (\x+\xb, \yb) circle (1pt) node [anchor=north east] {$dhk$};
    \fill (\x+\xc, \yc) circle (1pt) node [anchor=south west] {$bgk$};
    \fill (\x+\xd, \yd) circle (1pt) node [anchor=south east] {$chk$};
    \fill (\x+\xe, \ye) circle (1pt) node [anchor=north west] {$bfk$};
    \fill (\x+\xf, \yf) circle (1pt) node [anchor=north east] {$dgk$};
    \fill (\x+\xa/2+\xd/2,\ya/2+\yd/2) circle (1pt) node  [anchor=south east] {$ck$};
    \fill (\x+\xb/2+\xd/2,\yb/2+\yd/2) circle (1pt) node  [anchor=east] {$hk$};
    \fill (\x+\xb/2+\xf/2,\yb/2+\yf/2) circle (1pt) node  [anchor=north east] {$dk$};
    \fill (\x+\xf/2+\xc/2,\yf/2+\yc/2) circle (1pt) node [anchor=south east] {$gk$};
    \fill (\x+\xe/2+\xc/2,\ye/2+\yc/2) circle (1pt) node [anchor=west] {$bk$};
    \fill (\x+\xe/2+\xa/2,\ye/2+\ya/2) circle (1pt) node [anchor=north east] {$fk$};
    }
    \foreach \x in {\xe}
    {
    \fill (\x+\xa, \ya) circle (1pt) node [anchor=south west] {$bfi$};
   \fill (\x+\xc, \yc) circle (1pt) node [anchor=west] {$bej$};
   \fill (\x+\xe, \ye) circle (1pt) node [anchor=west] {$bei$};
    \fill (\x+\xf, \yf) circle (1pt) node [anchor=north west] {$bgj$};
    \fill (\x+\xa/2+\xd/2,\ya/2+\yd/2) circle (1pt) node  [anchor=north west] {$bf$};
   \fill (\x+\xb/2+\xf/2,\yb/2+\yf/2) circle (1pt) node  [anchor=south west] {$bg$};
    \fill (\x+\xf/2+\xc/2,\yf/2+\yc/2) circle (1pt) node [anchor=north west] {$bj$};
    \fill (\x+\xe/2+\xc/2,\ye/2+\yc/2) circle (1pt) node [anchor=west] {$be$};
    \fill (\x+\xe/2+\xa/2,\ye/2+\ya/2) circle (1pt) node [anchor=south west] {$bi$};
    }
  \foreach \y in {-\ya}
   {
     \fill (-\xe/2,\y-\ye/2) circle (1pt) node [anchor=north east] {$dg$};
     \fill (0,\y-\ye) circle (1pt) node [anchor=north] {$dgj$};
     \fill (\xe/2,\y-\ye/2) circle (1pt) node [anchor=north west] {$gj$};
     \fill (0,\y) circle (1pt) node [above left=0.5ex and -0.5ex] {$g$};
}
   \foreach \y in {\ya}
   {
     \fill (-\xe/2,\y+\ye/2) circle (1pt) node [anchor=south east] {$cf$};
     \fill (0,\y+\ye) circle (1pt) node [anchor=south] {$cfi$};
     \fill (\xe/2,\y+\ye/2) circle (1pt) node [anchor=south west] {$fi$};
     \fill (0,\y) circle (1pt) node [below left=0.5ex and 0ex] {$f$};
 }
\end{tikzpicture}
\end{center}
Here 
a vertex labeled $x$ represents the necklace with flanked flag $(\Delta^1 \vee \Delta^1, \{a,x,l\})$, 
a vertex labeled $xy$ represents the necklace with flanked flag $(\Delta^1 \vee \Delta^1 \vee \Delta^1, \{a,x,y,l\})$, and
a vertex labeled $xyz$ represents the necklace with flanked flag $(\Delta^1 \vee \Delta^1 \vee \Delta^1 \vee \Delta^1, \{a,x,y,z,l\})$. 
A triangle between vertices $x$, $xy$ and $xyz$ represents the necklace with flanked flag $( \Delta^2 \vee \Delta^2, \{a,x,l\} \subset \{a,x,y,l\} \subset \{a,x,y,z,l\})$.

Thus $\Map_X(a,l)$ is homotopy equivalent to $S^1 \vee S^1$, the wedge of two circles. 
\end{example}

\section{Future directions}
\label{sec:futureWork}

It would be very nice to have a Quillen model structure on simplicial categories appropriate to their use as models for concurrency. A natural candidate is J.\ Bergner's model structure on simplicial categories~\cite{bergner:scat}, which is closely related to A.\ Joyal's quasi-category model structure on simplicial sets~\cite{joyal:qcatAndKan}. However this model structure is too strong, since the weak equivalences induce equivalences of path categories. Thus a weaker notion of equivalence is needed. For some recent ideas in this direction, see~\cite{krishnan:cubicalApproximation}. 

The simplicial set models in this paper are in fact triangulations of
cubical set models. When constructing models of the execution spaces
from one state to another, we used (simplicial) necklaces on the
simplicial models. Instead, one could use cubical necklaces on the
cubical models, if one understood such things from a homotopy-theoretic
point of view. These models would be more economical. For example, we should be able to redo Example~\ref{example:petriNet} using cubical sets and cubical necklaces to obtain the following cubical set for $\Map_X(a,l)$.

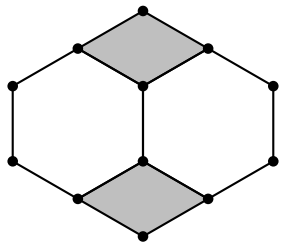
\begin{center}
 \begin{tikzpicture}[thick,scale=1]
    \pgfmathparse{cos(90)}
     \let\xa\pgfmathresult
     \pgfmathparse{sin(90)}
     \let\ya\pgfmathresult
     \pgfmathparse{cos(150)}
     \let\xd\pgfmathresult
     \pgfmathparse{sin(150)}
     \let\yd\pgfmathresult
     \pgfmathparse{cos(210)}
     \let\xb\pgfmathresult
     \pgfmathparse{sin(210)}
     \let\yb\pgfmathresult
     \pgfmathparse{cos(270)}
     \let\xf\pgfmathresult
     \pgfmathparse{sin(270)}
     \let\yf\pgfmathresult
     \pgfmathparse{cos(330)}
     \let\xc\pgfmathresult
     \pgfmathparse{sin(330)}
     \let\yc\pgfmathresult
     \pgfmathparse{cos(30)}
     \let\xe\pgfmathresult
     \pgfmathparse{sin(30)}
     \let\ye\pgfmathresult
    \foreach \y in {-\ya,\ya}
    {
      \draw [fill=lightgray] (0,\y-\ya+\ye) -- (\xe,\y) -- (0,\y+\ya-\ye) -- (-\xe,\y) -- (0,\y-\ya+\ye);
  }
     \foreach \x in {-\xe,\xe}
     {
   \draw (\x+\xa,\ya) -- (\x+\xd,\yd) -- (\x+\xb,\yb) -- (\x+\xf,\yf) -- (\x+\xc,\yc) -- (\x+\xe,\ye) -- (\x+\xa,\ya);
   }
    \foreach \x in {-\xe}
    {
    \fill (\x+\xa, \ya) circle (2pt); 
   \fill (\x+\xb, \yb) circle (2pt); 
    \fill (\x+\xc, \yc) circle (2pt); 
    \fill (\x+\xd, \yd) circle (2pt); 
    \fill (\x+\xe, \ye) circle (2pt); 
    \fill (\x+\xf, \yf) circle (2pt); 
   }
    \foreach \x in {\xe}
    {
    \fill (\x+\xa, \ya) circle (2pt); 
   \fill (\x+\xc, \yc) circle (2pt); 
   \fill (\x+\xe, \ye) circle (2pt); 
    \fill (\x+\xf, \yf) circle (2pt); 
   }
  \foreach \y in {-\ya}
   {
    \fill (0,\y-\ye) circle (2pt); 
}
   \foreach \y in {\ya}
   {
    \fill (0,\y+\ye) circle (2pt); 
}
\end{tikzpicture}
\end{center}

\section*{Acknowledgments}

The author would like to thank David Spivak for many useful conversations pertaining to this article. 
He would also like to thank the NSF for funding some of this research through the Midwest Topology Network (NSF grant DMS-0844249). 
In addition he would like to thank Martin Raussen and Lisbeth Fajstrup for organizing the GETCO 2010 conference that motivated this work.


\end{document}